\documentclass[draft, onecolumn, 11pt]{IEEEtran}
\IEEEoverridecommandlockouts

\usepackage[utf8]{inputenc} 
\usepackage[T1]{fontenc}
\usepackage{cite}
\usepackage[cmex10]{amsmath} 
\usepackage{amssymb,amsthm}
\usepackage[final]{hyperref}
\newtheorem{theorem}{Theorem}[]
\newtheorem{definition}[theorem]{Definition}

\newtheorem{lemma}[theorem]{Lemma}

\newtheorem{corollary}[theorem]{Corollary}
\newtheorem{remark}[theorem]{Remark}

\newtheorem{construction}[theorem]{Construction}

\newcommand{\F}{\mathbb{F}}

\usepackage{cite}
\usepackage{amsmath,amssymb,amsfonts}
\usepackage{algorithmic}
\usepackage{graphicx}
\usepackage{textcomp}
\usepackage{bm}
\usepackage{xcolor}
\def\BibTeX{{\rm B\kern-.05em{\sc i\kern-.025em b}\kern-.08em
    T\kern-.1667em\lower.7ex\hbox{E}\kern-.125emX}}

\begin{document}

\title{A Generalized Half-Singleton Bound for List Decoding of Insertion/Deletion Errors\\
\thanks{This work was supported by the MSCA Doctoral Networks 2021, HORIZON-MSCA-2021-DN-01 (ENCODE, grant \#101072316).} 
}

\author{\IEEEauthorblockN{{Neehar Verma, Camilla Hollanti, Razane Tajeddine}}\\
\IEEEauthorblockA{
\textit{Department of Mathematics and Systems Analysis, Aalto University,
Finland}\\
neehar.verma@aalto.fi, camilla.hollanti@aalto.fi, rstajeddine@gmail.com}
}

\maketitle

\vspace{-1.5cm}
\begin{abstract}
We derive an upper bound on the rate of a linear code capable of list decoding from insertion and deletion errors. Subsequently, we show that a linear code can list decode from at most $n - \frac{L}{L+1}k + 1$ insdel errors with fixed list size $L$, and that this bound is strict when the code does not contain the all-ones vector. We also provide a simple construction of codes achieving the optimal rate for any given list size capable of correcting a single insdel error. All the results hold for linear codes over finite fields of size $q \geq 3$.   
\end{abstract}

\section{Introduction}
Error-correcting codes are imperative in the design of reliable data storage and communication systems. Recently, we have witnessed the advent of efficient DNA synthesis and sequencing technologies which have enabled the development of efficient DNA-based archival data storage \cite{goldman2013towards}. The processes of synthesis and sequencing are inherently susceptible to insertion/deletion (\emph{insdel}) errors \cite{heckel2019characterization}. A deletion involves the removal of a symbol from an unknown position while an insertion involves the addition of a symbol at an unknown position. In contrast to typical substitutions and erasures, insdel errors destroy the alignment between transmitted and received coordinates and may even alter the length of the transmitted message. The application of coding theory to DNA-based storage has motivated renewed interest in codes capable of correcting insdel errors;    \cite{haeupler2021synchronization, milenkovic2024review, sabary2024survey} provide recent surveys on the topic. 
For a complementary information-theoretic perspective the reader may refer to \cite{shomorony2022information}. 

Classical codes for the Hamming metric have an established theory of bounds, constructions and decoding algorithms. 
The well-known \emph{Singleton bound} establishes constraints on the parameters of a code. List decoding relaxes unique decoding by allowing a bounded list of candidate codewords, enabling recovery beyond the unique decoding radius. We refer the reader to \cite{kumar2026advances} for a comprehensive survey on advances in list decoding. The relaxation on unique decoding permits parameters beyond those allowed by the Singleton bound.  The relaxed constraints are given by the \emph{generalized Singleton bound} by Shangguan and Tamo \cite{shangguan2023generalsingleton}. 

\subsection{Unique decoding of insdel errors}
When considering codes not necessarily constrained by linearity, a substantial amount of theory has been developed for unique decoding against insdel errors.
Following the seminal work of Levenshtein \cite{levenshtein1966binary} and the subsequent fundamental bounds of \cite{levenshtein2002bounds}, there has been staggering progress in the development of bounds \cite{bukh2017improved, kulkarni2013nonasymptotic, fazeli2015spherepacking, chen2021improved, kong2025combinatorial, kiah2026constantlist} to better understand the capabilities, limitations, and rate-distance tradeoff for the insdel model. To complement these bounds, there are several works that provide optimal constructions of codes for insdel errors \cite{guruswami2017deletion, sima2020optimal, cheng2018deterministic, guruswami2016efficiently, qiu2026binary}. It was shown in \cite{haeupler2021synchronization-org} that there exist explicit and efficiently decodable codes that correct a $\delta$-fraction of insdel errors while achieving rate $1-\delta-\varepsilon$, essentially attaining optimality with respect to the Singleton bound.   

For linear codes, stronger rate restrictions apply. In particular, it was shown in \cite{abdel2007linear} that a linear code capable of correcting a single deletion can have a rate at most $1/2$, far from the optimal rate of one achievable by non-linear codes. Extending these results, Cheng, Guruswami, Haeupler and Li \cite{cheng2020efficient} established the alphabet-independent \emph{half-Singleton bound} which says that any linear code that can correct a $\delta$ fraction of insdel errors has a rate at most $\frac{1-\delta}{2}+o(1)$. In other words, an $[n,k]$ linear code can correct at most $n-2k+1$ insdel errors. For linear codes not containing the all-ones vector, \cite{ji2023strict} established that the half-Singleton bound is strict. Turning our attention to structured code families, \cite{con2023reed} established that random Reed--Solomon (RS) codes attain the half-Singleton bound over large enough fields. Subsequently, it was shown in \cite{con2024randomRS} that random RS codes approach the bound over substantially smaller linear-sized fields. This result was recently improved in \cite{verma2026ag} and \cite{guan2026random} where it was shown that there exist randomized families of structured algebraic geometry  codes approaching the bound over constant-sized fields. Several other works have further explored the capabilities of RS codes against insdel errors \cite{beelen2025fulllengthratehalf, con2023optimal, liu2024optimal, zhang2026perminsdel}.

\subsection{List decoding of insdel errors}
As is the case for classical codes, list decoding for insdels relaxes the unique decoding requirement and permits parameters beyond the specified bounds for unique decoding. Recent literature has primarily worked towards investigating list decoding of insdel errors for non-linear codes. Wachter-Zeh \cite{wachter-zeh2018listdecoding} initiated the study of Johnson-type bounds for insdel codes and investigated list decoding of Varshamov--Tenengolts codes \cite{varshamov1965codes}. The original bound was amended in \cite{hayashi2020listdecodability} and subsequently strengthened in \cite{liu2023lowerbound, yang2026improvedjohnson}. The list decoding capacity of binary codes for insertions is derived in \cite{con2026listdecoding}. 
For polynomially bounded list sizes, \cite{haeupler2018synclist, guruswami2020optimally} provide explicit constructions of insdel codes. For constant-sized lists, \cite{guruswami2017deletion, hayashi2020listdecodability} provide explicit constructions. In \cite{kiah2026constantlist} upper and lower bounds are developed for constant-sized lists, with an explicit dependence on the list size.

In the case of linear codes, little is known when it comes to list decoding against insdel errors. In \cite{banerjee2025decoding} a connection is established between classical list recoverablity and list decocadabity of insdel codes. Li, Gabrys and Farnoud \cite{li2026linearlist} demonstrate that the unique decoding rate restriction of \cite{abdel2007linear} can be circumvented by permitting lists whose sizes are polynomially bounded in $n$. In particular, they derive acheivability results for binary random linear codes and subsequently provide explicit constructions of binary linear codes (by leveraging code concatenation) with reasonably sized lists and rate approaching one.

\textbf{Our contributions.} In this work we focus on constant list sizes. Our results hold universally for all linear codes over fields of size $q \geq 3$, the binary case is more subtle and requires additional restrictions on the code parameters (see Remark \ref{binarycounterexample}). We prove that, for any fixed $L$, the maximum rate achievable by a linear code capable of list decoding against a single insdel error with list size $L$ is at most $\frac{L}{L+1}$, recovering the unique decoding rate restriction of \cite{abdel2007linear} when $L=1$. We subsequently extend this result to obtain an upper bound on the number of insdel errors that an $[n,k]$ linear code of length $n$ and dimension $k$ can list decode. Specifically, we derive the \emph{generalized half-Singleton bound} which states that an $[n,k]$ linear code can correct at most $n - \frac{L}{L+1}k + 1$ insdel errors. Additionally we show that the generalized half-Singleton bound is strict when the code does not contain the all-ones vector. To complement these bounds we provide a simple construction of linear codes attaining the maximum possible rate/distance tradeoff.  

\section{Notation and preliminaries}
In the entirety of this work we use $\F_q$ to denote the finite field with $q$ elements, and we write $[n] = \{1, \dots, n\}$. We use $(\F_q)^*$ to denote the set of all words of arbitrary length over the alphabet $\F_q$.
An $[n,k,d]_q$ linear code $\mathcal{C}$ is defined as a subspace of $\F_q^n$, parameterized by its length $n$, dimension $k$, and minimum Hamming distance $d$. These three parameters of a code are related by the well-known Singleton bound, $d \leq n - k + 1$. Codes meeting this bound are called maximum distance separable (MDS) codes. We may simply refer to an $[n,k]$ code, when the field size and minimum distance are either not relevant or immediate from the context.

We recall some basic notions and relevant results concerning insdel codes. First we define the edit distance which serves as the fundamental quantity for measuring the distance between two strings. 

\begin{definition}[Edit distance]
	Given two strings $s$ and $s'$, the edit distance $d_{\mathrm{edit}}(s,s')$ is the minimum number of insertions/deletions required to convert $s$ to $s'$ (or vice-versa). 
\end{definition}

Next, we recall the consequence that linearity imposes on the rate of a code that can correct from a single insertion or deletion. This consequence was established in~\cite{abdel2007linear}.
\begin{lemma}
	An $[n,k]$ linear code $C$ with rate $\frac{k}{n} > \frac{1}{2}$ cannot correct a single insertion/deletion.
\end{lemma}

A direct consequence of the above lemma gives us the well-established non-asymptotic version of the half-Singleton bound which was proved in \cite{cheng2020efficient}.

\begin{theorem}[Non-asymptotic half-Singleton bound]
    An $[n,k]$ linear code can correct at most $n-2k + 1$ insdel errors.
\end{theorem}

Lastly, we define list decoding in the context of the edit distance. Let $B(y,r) = \{x\in (\F_q)^* \mid d_\mathrm{edit}(x,y) \leq r\}$ denote the edit ball around $y$ of radius $r$. 
\begin{definition}
    Let $\delta \in [0,1]$ and $L$ be a positive integer. A linear code $\mathcal{C} \subseteq \F_q^n$ is $(\delta,L)$-list decodable if for all $y \in (\F_q)^*$, $|B(y, \delta n) \cap \mathcal{C}| \leq L$.
\end{definition}

\section{Generalized half-Singleton bound}

In this section we determine the maximum number of insdel errors that a linear code is capable of correcting with a fixed list size $L$. To arrive at this bound we first prove a general version of the rate restriction in \cite{abdel2007linear}. In particular, we show that a linear code with rate greater than $\frac{L}{L+1}$ cannot list decode a single insertion/deletion with list size $L$. To prove this we want to show that there exist $L+1$ distinct codewords that can be transformed into some common received word by a single deletion. That is, we want to find $L+1$ distinct messages $x_0, \dots, x_L \in \F_q^k$ and distinct deletion positions $i_0, \dots, i_L \in [n]$ such that $x_0G^{(i_0)} = \cdots = x_LG^{(i_L)}$ where $G^{(i)}$ denotes the generator matrix with the $i^{th}$ column deleted. 
Importantly, we need these punctured generator matrices to maintain injectivity. To do this, for any given linear code $\mathcal{C}$, we denote $B = \{i \in [n] \mid e_i \in \mathcal{C}\}$ and restrict our deletion positions to the set $[n]\setminus B$. 

Formally, we will concern ourselves with the following definition of solution spaces $S^I$ and argue that for some increasing ordered tuple $I \in [n]^{L+1}$ there exists a vector $(x_0, \dots, x_L) \in S^I$ with pairwise distinct entries under the assumption that $\frac{k}{n}>\frac{L}{L+1}$.

\begin{definition}
    Let $\mathcal{C}$ be an $[n,k]$ linear code with generator matrix $G$, and let $\ell$ be a positive integer. For an increasing ordered tuple $I = (i_0, \dots, i_\ell) \in [n]^{\ell+1}$ we define the vector space
    $$\scalebox{0.9}{$S^I=\{(x_0, \dots, x_\ell) \in \F_q^{k(\ell+1)} \mid x_0G^{(i_0)} = \cdots = x_\ell G^{(i_\ell)} \in Y\}$},$$
    where $Y$ is some subspace of $\F_q^{n-1}$.  
    We also define the following chain of subspaces $S^I = S^I_0 \supseteq S^I_1 \supseteq \cdots \supseteq S_\ell^I$ where 
    $$S^I_j = \{(x_0, \dots, x_\ell) \in S^I \mid x_0 = \cdots = x_j\}.$$
\end{definition}

Since we will be restricting our deletion positions to the set $[n]\setminus B$ we will also need to restrict our received word space to be amicable to the proof technique. We set this restricted word space to be 
$$Y = \bigcap_{i\in B} \{xG^{(i)}\} \cap \{y \mid y_{i-1} = y_{i} , i\in B\setminus\{1,n\}\} \subseteq \F_q^{n-1}.$$
The dimension of this restricted word space is bounded from below by $\dim(Y) \geq n-1 - |B|(n-k) - |B\setminus\{1,n\}|$. 
Note that the vector space $S^I$ consists of vectors in $\F_q^{k(\ell+1)}$ imposed with $\ell(n-1)$ linear constraints. This along with the additional constraints imposed by the restriction to $Y$ gives us a simple lower bound on the dimension 
$$\dim(S^I) \geq k(\ell+1) - \ell(n-1) - |B|(n-k) - |B\setminus\{1,n\}|.$$

Considering any increasing tuple $I \in [n]^{\ell + 1}$ that begins with $a = \min([n] \setminus B)$ and ends with $z = \max([n]\setminus B)$ we also have 
$$\dim(S^I_\ell) \leq 1 + |B \cap \{1,n\}|$$
since a vector in $S_\ell^I$ is identified by a single vector $x \in \F_q^k$ with the equality $xG^{(a)} = xG^{(z)}$ forcing $xG$ to be constant from $a$ through $z$. The restriction to the space $Y$ also forces $xG$ to be separately constant from $1$ through $a-1$ and also from $z+1$ through $n$, therefore the dimension is determined by at most $3$ symbols of $\F_q$ contingent on the the size of $B\cap\{1,n\}$ as desired.

Now, to arrive at the pairwise distinctness of the entries in some vector of the solution space we first prove the following lemma, which establishes that under our high rate assumption there exists a tuple $I \in [n]^{L - 2|B| +1}$ of puncturing positions such that the inclusions in the chain $S^I = S^I_0 \supseteq S^I_1 \supseteq \cdots \supseteq S_\ell^I$ are strict.  

\begin{lemma} \label{strictinclusion}
Let $\mathcal{C}$ be an $[n,k]$ linear code with generator matrix $G$. Let $L \geq 2|B| + 1$. If $1 >\frac{k}{n} > \frac{L}{L+1}$, then there exists an increasing tuple of puncturing coordinates $I \in ([n]\setminus B)^{L - 2|B| +1}$ such that
$$S^I = S^I_0 \supsetneq S^I_1 \supsetneq \cdots \supsetneq S_{L - 2|B|}^I$$
with $\dim(S^I_j) - \dim(S_{j+1}^I) = 1$ for all $1 \leq j \leq L- 2|B| -1$.
\end{lemma}
\begin{proof}
We will prove the statement via induction. Initially consider the tuple $I_0 = (a,z)$ where $a = \min([n] \setminus B)$ and $z = \max([n]\setminus B)$. For this tuple we have $\dim(S^{I_0}_0) \geq 2k - (n-1) - |B|(n-k) - |B\setminus\{1,n\}|$ and $\dim(S^{I_0}_1) \leq 1 + |B \cap \{1,n\}|$. Therefore we have
\begin{align*}
\dim(S^{I_0}_0) - \dim(S^{I_0}_1) & \geq 2k - n - |B|(n-k) - |B|\\
& = k - (|B|+1)(n-k) - |B| \\
& > (L - |B| - 1)(n-k) - |B| \geq 0
\end{align*}
which directly implies $S^{I_0}_0 \supsetneq S^{I_0}_1.$

Now, assume that for some $1 \leq r \leq L-2|B| - 2$ there exists an increasing tuple $I_r = (a, i_1, \dots, i_r, z) \in ([n]\setminus B)^{r + 2}$  such that $S^{I_r} = S^{I_r}_0 \supsetneq S^{I_r}_1 \supsetneq \cdots \supsetneq S_{r + 1}^{I_r}$ with $\dim(S^{I_r}_j) - \dim(S_{j+1}^{I_r}) = 1$ for all $1 \leq j \leq r$.

We can write the difference $\dim(S^{I_r}_1) - \dim(S_{r+1}^{I_r})$ as a telescopic series. Under our hypothesis each term of the series is equal to $1$,  giving us 
$$\dim(S^{I_r}_1) - \dim(S_{r+1}^{I_r}) = \sum_{j = 1}^{r} (\dim(S_j^{I_r}) - \dim(S_{j+1}^{I_r})) = r.$$ 
Therefore $\dim(S_1^{I_r}) \leq r + 1 + |B\cap\{1,n\}|$, and
\begin{align*}
    \dim(S_0^{I_r}) - \dim(S_1^{I_r}) &  \geq k - (|B| + r + 1)(n-k) - |B|  \\
    & > (L - |B| - r - 1)(n-k) - |B| \\
    & \geq 1.
\end{align*}

Now, consider the following intermediate subspaces $T_a, T_{a+1}, \dots, T_{i_1}$ where $$T_{\varphi} = \{(x_0, \dots, x_{r+1}) \in S^{I_r}_0 \mid x_0G^{(a)} = \cdots = x_0G^{(\varphi)}\}.$$ Clearly $S_0^{I_r} = T_a \supseteq T_{a+1} \supseteq \cdots \supseteq T_{i_1} = S_1^{I_r}$. Furthermore, since each step imposes at most $1$ new linear constraint, we have $\dim(T_\varphi) - \dim(T_{\varphi+1}) \leq 1$ for all $ a \leq \varphi \leq i_1 -1$. If $\varphi \in B$ then any additional linear constraints disappear due to our definition of the restricted word space $Y$. Since $\dim(S_0^{I_r}) - \dim(S_1^{I_r}) > 1$ there are at least two strict inclusions in this chain. Therefore there exists some index $i_{\mathrm{new}} \in [n]\setminus B$ such that $S_0^{I_r} \supsetneq T_{i_{\mathrm{new}}} \supsetneq S_1^{I_r}$
with $\dim(T_{i_{\mathrm{new}}}) - \dim(S_1^{I_r}) = 1$.

Next we define the tuple $I_{r+1} = (a, i_{\mathrm{new}}, i_1, \dots, i_r, z)$. The subspace $S_1^{I_{r+1}}$ is constrained by $x_0 = x_1$, deleting the duplicate entry gives us exactly $T_{i_\mathrm{new}}$.
Similarly for any $1 \leq j \leq r+1$ the subspace $S_{j+1}^{I_{r+1}}$ with the newly inserted coordinate deleted is isomorphic to $S_{j}^{I_{r}}$. Therefore,
$$\dim(S_1^{I_{r+1}}) - \dim(S_{2}^{I_{r+1}}) = \dim(T_{i_{\mathrm{new}}}) - \dim(S_{1}^{I_{r}}) = 1.$$ 
and by the induction hypothesis,
$$\dim(S_j^{I_{r+1}}) - \dim(S_{j+1}^{I_{r+1}}) = \dim(S_{j-1}^{I_{r}}) - \dim(S_{j}^{I_{r}}) = 1$$
for all $2 \leq j \leq r+1$.

It remains to show that $S_0^{I_{r+1}} \supsetneq S_1^{I_{r+1}}$. We have
$$\begin{aligned}
    \dim(S_1^{I_{r+1}}) = \dim(T_{i_\mathrm{new}}) & = \dim(S_1^{I_r}) + 1 \\
    &\leq r + 2 +  |B\cap\{1,n\}|.
\end{aligned}$$
Therefore we get
$$\scalebox{0.99}{$\begin{aligned}
    \dim(S_0^{I_{r+1}}) - \dim(S_1^{I_{r+1}}) & \geq k - (|B| + r + 2)(n-k) - |B|  \\
    & > (L - |B| - r - 2)(n-k) - |B| \\
    & \geq 0.
\end{aligned}$}$$
Hence $S_0^{I_{r+1}} \supsetneq S_1^{I_{r+1}}$. Therefore we have 
$$S^{I_{r+1}}_0 \supsetneq S^{I_{r+1}}_1 \supsetneq \cdots \supsetneq S_{r + 2}^{I_{r+1}}$$ 
with $\dim(S^{I_{r+1}}_j) - \dim(S_{j+1}^{I_{r+1}}) = 1$ for all $1 \leq j \leq r + 1$. This completes the induction.
\end{proof}

\begin{theorem} \label{ratesingleinsdel}
Let $\mathcal{C}$ be an $[n,k]$ linear code over $\F_q$ with $q \geq 3$. If $1>\frac{k}{n}>\frac{L}{L+1}$, then $\mathcal{C}$ is not $(1/n, L)$-list decodable.
\end{theorem}

\begin{proof}
If $L\leq 2|B|$, then deleting from the zero codeword and from the nonzero multiples of $e_i$, for every $i\in B$, gives $1+|B|(q-1) \geq 1+2|B| > L$ distinct codewords deleting to the zero word. 

Henceforth we assume $L\geq 2|B| + 1$. By Lemma \ref{strictinclusion}, there exists an increasing tuple of puncturing coordinates $I=(i_0,\ldots,i_{L-2|B|})\in([n]\setminus B)^{L-2|B|+1}$
such that
$$S^I=S_0^I\supsetneq S_1^I\supsetneq\cdots
\supsetneq S_{L-2|B|}^I,$$
with $\dim(S_j^I)-\dim(S_{j+1}^I)=1$ for all $1\le j\le L-2|B|-1$.

Now, suppose for the sake of contradiction that every vector in $S^I$ has $x_a=x_b$ for some $a<b$. This implies
$$
x_a=x_{a+1}=\cdots=x_b,
$$
by the fact that $x_aG$ is constant on $[i_a,i_b]$ and puncturing at each coordinate in $I$ is injective.

Let $j$ be the smallest index such that every vector in $S^I$ has a repeated vector amongst $(x_0,\ldots,x_j)$. Fix some
$x=(x_0,\ldots,x_{L-2|B|})\in S^I$ with distinct $x_0,\ldots,x_{j-1}$. Its repetition must be $x_{j-1}=x_j$.

Now, for any $z\in S_{j-1}^I$, the vector $x+z\in S^I$ has its first $j$ entries distinct, since the first $j$ distinct entries of $x$ are only translated by a constant. Therefore, by the minimality of $j$, we have
$$x_j+z_j=x_{j-1}+z_{j-1}.$$
Since $x_j=x_{j-1}$, we obtain $z_j=z_{j-1}$. Thus $z\in S_j^I$, and hence $S_{j-1}^I=S_j^I$ which
contradicts the strict inclusion $S_{j-1}^I\supsetneq S_j^I$.

We therefore must have $L-2|B|+1$ distinct codewords deleting to some common received word $y\in Y$. By the definition of $Y$ we have $q-1$ additional distinct codewords that delete to $y$ for every $i \in B$.
In total, we have at least
$$L - 2|B| + 1 + |B|(q-1) = L + 1 + |B|(q - 3) \geq L + 1$$
distinct codewords deleting to $y$.
\end{proof}

\begin{corollary}[Non-asymptotic generalized half-Singleton bound]\label{cor:nonasy}
Let $\mathcal{C}$ be an $[n,k]$ linear code over $\F_q$ with $q \geq 3$ that is $(t/n, L)$-list decodable. If $t \leq n - k$ then $$t \leq n - \frac{L+1}{L}k + 1.$$
\end{corollary}

\begin{proof}
Suppose the $[n,k]$ linear code $\mathcal C$ can list decode $t$ insdel errors with list size $L$, where $1\le t\le n-k$. Puncture $\mathcal C$ at any $t-1$ positions outside an information set to give $\mathcal C'$. Then $\mathcal C'$ is an $[n-t+1,k]$ linear code, with $n-t+1>k$.

If $L+1$ distinct codewords of $\mathcal C'$ were within one insdel error of the same received word, their preimages in $\mathcal C$ would be within $t$ insdel errors of that word, contradicting the assumed list size $L$ for $\mathcal{C}$. This implies that $\mathcal C'$ can list decode one insdel error with list size $L$.
By Theorem \ref{ratesingleinsdel} we have that
$\frac{k}{n-t+1} \leq \frac{L}{L+1}$,
and therefore $$t\le n-\frac{L+1}{L}k+1.$$
\end{proof}

\begin{remark} \label{binarycounterexample}
Note that Theorem \ref{ratesingleinsdel} and Corollary \ref{cor:nonasy} are valid over fields of size $q \geq 3$. The binary case is more subtle and requires additional restrictions on the code parameters. If the binary code has Hamming distance $d > 1$, then the proof of Theorem \ref{ratesingleinsdel} goes through. Consequently, if $d > t$, then the proof of Corollary \ref{cor:nonasy} goes through. Otherwise, for every $L \geq 3$, the following binary single parity check code with $d = 1$ provides a valid counter example:
$$\mathcal{C}_L = \{c \in \F_2^{2L-1} \mid c_1 + c_3 + \cdots + c_{2L-1} = 0\}.$$
The code $\mathcal{C}_L$ is $(1/(2L-1),L)$-list decodable. The rate of the code is $\frac{2L-2}{2L-1} > \frac{L}{L+1}$, going past the bound in Theorem \ref{ratesingleinsdel}.
\end{remark}

While Theorem \ref{ratesingleinsdel} and Corollary \ref{cor:nonasy} may not hold in the binary case, we have the following looser bound which holds for any linear code.
\begin{remark}\label{loosebound}
Without the restriction $t\le n-k$ and $q\geq 3$, we have the looser bound
$$t\le n-\frac{L+1}{L}(k-\lfloor\log_q L\rfloor)+1,$$
which holds for all $q \geq 2$.
If $t>n-k$, retaining $n-t+1$ linearly independent generator columns gives the full space $\mathbb F_q^{n-t+1}$. The zero word and the nonzero multiples of its unit vectors delete to the zero word of length $n-t$. Each has $q^{k-n+t-1}$ preimages in the original code, and thus
$$q^{k-n+t-1}(1+(n-t+1)(q-1)) \leq L.$$
The looser bound follows from this inequality without having to invoke Theorem \ref{ratesingleinsdel}.
\end{remark}

Asymptotically, as $n$ tends to infinity, this gives us the following bound

\begin{corollary}[Asymptotic generalized half-Singleton bound]\label{generalhalfsingleton}
Let $\mathcal{C}$ be an $[n,k]$ linear code that is $(\delta,L)$-list decodable, then it has rate at most $\frac{L}{L+1}(1 - \delta) + o(1)$. 
\end{corollary}

The proofs of Theorem \ref{ratesingleinsdel} and Corollary \ref{cor:nonasy} implicitly imply that the upper bound can only be attained by codes containing the all-ones vector. In the following we explicitly show that these bounds are strict when $\bm{1} \notin \mathcal{C}$.
\begin{theorem} \label{ratesingleinsdelno1}
Let $\mathcal{C}$ be an $[n,k]$ linear code over $\F_q$ with $q \geq 3$. Assume that $\bm{1} \notin \mathcal{C}$. If $1>\frac{k}{n}\geq\frac{L}{L+1}$, then $\mathcal{C}$ is not $(1/n,L)$-list decodable.
\end{theorem}
\begin{proof}
    The proof follows identically to that of Lemma \ref{strictinclusion} and Theorem \ref{ratesingleinsdel} with one modification. The extra condition that $\bm{1} \notin \mathcal{C}$ gives us a tighter bound on the dimension of the final subspace for each step of the induction, that is, we get  
    $$\dim(S^I_\ell) \leq |B \cap \{1,n\}|$$
    instead of $\dim(S^I_\ell) \leq 1 + |B \cap \{1,n\}|$. This widens the gap in dimension between the first and last subspaces in the chain by one. The proof concludes identically to that of Lemma \ref{strictinclusion}, with the additional dimension permitting the existence of an increasing tuple $I$ even when $\frac{k}{n} = \frac{L}{L+1}$. The remainder of the proof is then identical to that of Theorem \ref{ratesingleinsdel}.
\end{proof}

\begin{corollary}[Non-asymptotic generalized strict half-Singleton bound]
Let $\mathcal{C}$ be an $[n,k]$ linear code over $\F_q$ with $q \geq 3$ that is $(t/n,L)$-list decodable. Assume $\bm{1} \notin \mathcal{C}$. If $t \leq n - k$ then $$t < n - \frac{L+1}{L}k + 1.$$
\end{corollary}
\begin{proof}
    The proof is identical to that of Corollary \ref{cor:nonasy}, the only difference being the invocation of Theorem \ref{ratesingleinsdelno1} instead of Theorem \ref{ratesingleinsdel}.
\end{proof}

\section{Optimal construction}
In this section we provide explicit constructions of codes over any finite field $\F_q$ with $q \geq 3$ capable of list decoding from a single insdel for any given fixed list size $L$. These codes attain the bound in Theorem \ref{ratesingleinsdel} and therefore show that it is tight. 

\begin{construction}
    Let $\F_q$ be a finite field with $q\geq 3$ and let $\alpha \notin \{0,1\}$. For any integer $L \geq 2$, the $[2(L+1),2L]$ code generated by the matrix
    $$\scalebox{0.9}{$G = \begin{bmatrix}
        \alpha & 1 &  1 & 0 & 0 & 0 & 0 & \cdots& 0 \\
        0 & \alpha - 1 &  0 & 1 & 0 & 0 & 0 &\cdots & 0 \\
        0 & 0 &  \alpha - 1 & 0 & 1 & 0 & 0 & \cdots & 0 \\
        0 & \cdots & 0 & \ddots & \ddots & \ddots & 0 & \cdots & 0 \\
        0 & \cdots & 0 & 0 & \alpha - 1 & 0 & 1 & 0 & 0 \\
        0 & \cdots & 0 & 0 & 0 & \alpha - 1 & 0 & 1 & 0 \\
        0 & \cdots& 0 & 0  & 0 & 0 & \alpha - 1 & \alpha - 1 & \alpha
    \end{bmatrix}$}$$
    is $(1/(2L + 1),L)$-list decodable.
\end{construction}
\begin{proof}
    Suppose for the sake of contradiction that the code generated by the generator matrix $G$ cannot list decode from a single deletion error with list size $L$. That is, there exist $L+1$ distinct messages $x_0, \dots, x_L \in \F_q^{2L}$ and an increasing tuple of deletion coordinates $I = (i_0, \dots, i_L) \in [2(L+1)]^{L+1}$ such that 
    $$x_0 G^{(i_0)} = \cdots = x_L G^{(i_L)}.$$
    Suppose there exists some $a$ such that the adjacent deletion coordinates $i_a$ and $i_{a +1}$ are consecutive, i.e., $i_a + 1 = i_{a + 1}$. We write $x_a = (x_{a,1}, \dots, x_{a,2L})$ and $x_{a + 1} = (x_{a+1, 1}, \dots, x_{a+1, 2L})$.
    Comparing coordinatewise, the condition $x_aG^{(i_a)} = x_{a+1}G^{(i_a + 1)}$ forces $x_a = x_{a+1}$, contradicting the distinctness of the messages.

    Therefore, $I$ cannot be such that two adjacent deletion coordinates differ by one. This leaves us with the following $L + 2$ possibilities for $I$, 
    $$I_j = (1, \cdots, 2j - 1, 2j + 2, \cdots, 2L+2), ~ \text{where} ~ 0 \leq j \leq L + 1.$$
    Consider the cases in which $2 \leq j \leq L+1$. The conditions $x_j G^{(2j+2)} = \cdots = x_{L} G^{(2L+2)}$ force $$x_{j,2j} + x_{j, 2j + 2} = 2x_{j, 2j + 1},$$
    while the conditions $x_0 G^{(1)} = \cdots = x_{j-1} G^{(2j-1)}$ force $$x_{j-1,2j-3} + x_{j-1, 2j - 1} = 2x_{j-1, 2j-2}.$$ 
    Along with the condition $x_{j-1} G^{(2j-1)} = x_j G^{(2j+2)}$ this gives us $\alpha(x_{j-1,2j-1} - x_{j-1, 2j-2}) = 0$. Since $\alpha \neq 0$ we get $x_{j-1,2j-1} = x_{j-1, 2j-2}$, and by the previous equality we get $$x_{j-1, 2j-3} = x_{j-1,2j-2} = x_{j-1, 2j-1}.$$
    Now, working backwards through the conditions $x_0 G^{(1)} = \cdots = x_{j-1} G^{(2j-1)}$ we arrive at $x_{1,1} = x_{1,2} = x_{1,3}$. The condition $x_0G^{(1)} = x_1G^{(3)}$ tells us that $x_{0,1} + (\alpha - 1)x_{0,2} = \alpha x_{1,1}$ and $x_{0,s} = x_{1,s}$ for all $s\neq 1$. Along with $x_{1,1} = x_{1,2} = x_{1,3}$ we get $x_{0,1} = x_{1,1}$ and therefore $x_0 = x_1$, contradicting the distinctness of the messages. 
    The proof for the cases in which $0 \leq j \leq L-1$ follows symmetrically. Therefore, this eliminates all increasing tuples of $L + 1$ deletion coordinates $I \in [2(L+1)]^{L+1}$, showing that the given code can correct one deletion error with list size at most $L$.
    
    Now, for insertions, suppose that the code generated by the generator matrix $G$ cannot list decode from a single insertion error with list size $L$. That is, there exist $L+1$ distinct messages $x_0, \dots, x_L \in \F_q^{2L}$ and an increasing tuple of insertion coordinates $I = (i_0, \dots, i_L) \in [2(L+1)+1]^{L+1}$ such that 
    $$x_a G = y^{(i_a)} \quad \text{for all } 0\leq a \leq L,$$
    for some common received word $y \in \F_q^{2(L+1)+1}$, where $y^{(i)}$ denotes $y$ with the $i^{th}$ coordinate deleted. We know that $x_a G - x_{a+1}G = y^{(i_a)} - y^{(i_{a+1})}$ is supported on the set $\{i_a, \dots, i_{a+1} - 1\}$. If $i_{a+1} - i_a = 2$, i.e. the support consists of two consecutive coordinates, then we would have $x_aG^{(i_a + 1)} = x_{a+1}G^{(i_a)}$. By our previous deletion argument for consecutive deletion coordinates we get $x_a = x_{a+1}$. If $i_{a+1} - i_a = 1$ then we would have $x_aG^{(i_a)} = x_{a+1}G^{(i_a)}$, which again gives us the contradiction $x_a = x_{a+1}$ since $G^{(i_a)}$ is full rank. Therefore it must be that $i_{a+1} - i_a \geq 3$ for all $0 \leq a \leq L-1$. This gives us $$2L + 3 \geq i_L = \sum_{a=0}^{L-1} (i_{a+1}-i_a) + i_0 \geq 3L + 1,$$
    which is impossible for any $L>2$. When $L=2$ the only remaining possibility for the insertion coordinates is $I = (1,4,7)$. The conditions $y^{(1)}, y^{(4)}, y^{(7)} \in \mathcal{C}$ force $y$ to be constant, and therefore $x_0 = x_1 = x_2$, a contradiction. 
\end{proof}

\section{Conclusion} In this paper, we provided an upper bound for the rate of a linear
code over $\F_q \ (q\geq 3)$ capable of list decoding from insertion and deletion errors. A code construction achieving the optimal rate for any list size capable of correcting a single insdel error was also given. In future work, it would be interesting to construct codes achieving the bounds in Corollaries~\ref{cor:nonasy} and~\ref{generalhalfsingleton}.

\bibliographystyle{ieeetr}
\bibliography{main}

\section*{AI usage}
ChatGPT 5.6 was used to check for grammatical errors and improve the readability of the article. It was also used for sanity checking proofs and to provide the binary counterexample in Remark \ref{binarycounterexample}. The authors take full responsibility for all claims made in this manuscript. 
\end{document}